\documentclass[letterpaper, 10 pt, conference]{ieeeconf}  %

\IEEEoverridecommandlockouts                              %

\usepackage{hyperref}%
\hypersetup{colorlinks=true, linkcolor=blue, breaklinks=true, urlcolor=blue}
\usepackage{doi}

\usepackage[all]{xy}
\usepackage{amsmath}
\usepackage{amscd}
\usepackage{mathrsfs}
\usepackage{amssymb}
\usepackage{tikz}
\usepackage[utf8]{inputenc}
\usepackage[yyyymmdd]{datetime}
\usepackage{graphicx}
\usepackage{subcaption}
\usepackage{lipsum}
\usepackage{multicol} \usepackage{mathtools}
\usepackage{multirow}
\usepackage{threeparttable}
\usepackage{pdfpages}
\usepackage{stmaryrd}
\usepackage{selectp}

\usepackage{multicol} \usepackage{mathtools}
\usepackage{version,xspace}
\usepackage{url,doi}

 \newcommand{\setdef}[2]{\{#1
	\; | \; #2\}}

\newcommand\oprocendsymbol{\hbox{$\triangle$}}
\newcommand\oprocend{\relax\ifmmode\else\unskip\hfill\fi\oprocendsymbol}

\let\labelindent\relax
                  
\usepackage{enumitem} \renewcommand{\theenumi}{(\roman{enumi})}

\DeclareSymbolFont{bbold}{U}{bbold}{m}{n}
\DeclareSymbolFontAlphabet{\mathbbold}{bbold}
\newcommand{\vect}[1]{\mathbbold{#1}}

\newcommand{\real}{\mathbb{R}}

\newcommand{\seminorm}[1]{{\left\vert\kern-0.25ex\left\vert\kern-0.25ex\left\vert #1
		\right\vert\kern-0.25ex\right\vert\kern-0.25ex\right\vert}}

\newcommand{\semimeasure}[1]{\mu_{\seminorm{\cdot}}\kern-0.5ex\left(#1\right)}

\newtheorem{theorem}{Theorem}[section]
\newtheorem{proposition}[theorem]{Proposition}

\newtheorem{definition}[theorem]{Definition}
	\newtheorem{remark}[theorem]{Remark}
	\newtheorem{example}[theorem]{Example}
	
        \newtheorem{assumption}[theorem]{Assumption}

\newcommand{\suchthat}{\;\ifnum\currentgrouptype=16 \middle\fi|\;}
\newcommand{\scirc}{\raise1pt\hbox{$\,\scriptstyle\circ\,$}}

\newcommand{\OF}{\mathsf{F}}
\newcommand{\OG}{\mathsf{G}}

\newcommand{\OT}{\mathsf{T}}

\title{\LARGE \bf
A Contracting Dynamical System Perspective toward Interval Markov Decision Processes
}

\author{Saber Jafarpour and Samuel Coogan%
\thanks{*This work is partially supported by the National Science Foundation under awards \#1749357, \#1931980, and \#1924978,  the Air Force Office of Scientific Research under Grant FA9550-23-1-0303, and   NASA under the University Leadership Initiative (ULI)  Award \#80NSSC20M0163 but solely reflects the opinions and conclusions of its authors.}%
\thanks{Saber Jafarpour and Samuel Coogan are with the School of Electrical and Computer Engineering, Georgia Institute of Technology, USA, {\tt\small \{saber,sam.coogan\}@gatech.edu}}%
}

\begin{document}

\maketitle
\thispagestyle{empty}
\pagestyle{empty}

\begin{abstract}
Interval Markov decision processes are a class of Markov models where the transition probabilities between the states belong to intervals.
In this paper, we study the problem of efficient estimation of the optimal policies in Interval Markov Decision Processes (IMDPs) with continuous action-space.   
Given an IMDP, we show that the pessimistic (resp. the optimistic) value iterations, i.e., the value iterations under the assumption of a competitive adversary (resp. cooperative agent), are monotone dynamical systems and are contracting with respect to the $\ell_{\infty}$-norm. 
Inspired by this dynamical system viewpoint, we introduce another IMDP, called the action-space relaxation IMDP.
We show that the action-space relaxation IMDP has two key features: (i) its optimal value is an upper bound for the optimal value of the original IMDP, and (ii) its value iterations can be efficiently solved using tools and techniques from convex optimization.  
We then consider the policy optimization problems at each step of the value iterations as a feedback controller of the value function. Using this system-theoretic perspective, we propose an iteration-distributed implementation of the value iterations for approximating the optimal value of the action-space relaxation IMDP. 

\end{abstract}

\section{Introduction}

\paragraph*{Motivation and Problem Statement}

Markov decision process (MDP) is a powerful and classical framework for modeling the stochastic interactions between a system and its %
environment~\cite{MLP:14}. The MDP framework has been successfully used to study various problems in dynamic decision-making~\cite{DB-JNT:96} and reinforcement learning~\cite{RSS-AGB:18}. 
A fundamental assumption in the MDP framework is that the parameters of the model are known or are learnable. However, in many real-world applications, the model parameters are typically estimated or inferred using data-driven methods and, thus, they are far from accurate. 
In the literature, several different approaches have been proposed to analyze MDPs with parameter uncertainties. %
In~\cite{GNI:05,AN-LEG:05}, robust dynamic programming is proposed to study optimal solutions of Markov decision processes with uncertainty in transition probabilities. In~\cite{SHQL-AA-PLG-BA:21}, a set-valued fixed-point equation is proposed to study the optimal value of Markov decision processes with uncertain reward functions. In~\cite{DT-GA-SC:14}, computationally efficient algorithms are developed to infer the unknown parameters in MDPs.

Interval Markov decision processes (IMDPs) are a class of Markov models with interval-bounded transition probabilities and reward functions~\cite{RG-SL-TD:00}. IMDPs can be considered as a family of MDPS and they appear naturally in the setting where systems are modeled using MDPs with uncertain parameters or with parameters obtained from data-driven sampling approaches. An alternative interpretation for the IMDP framework comes from a game-theoretic perspective. In this case, an IMDP models how an MDP interacts with the environment in the presence of an agent who resolves uncertain transition probabilities~\cite{GD-ML-MM-LL:22}.
In the literature, IMDPs have been used to analyze a various tasks including checking temporal logic specifications~\cite{EMW-UT-RMM:12} and motion planning in robotics~\cite{JJ-YZ-SC:22}.
Much of the early works on the IMDPs focus on models with finite number of actions~\cite{RG-SL-TD:00,GNI:05,AN-LEG:05}. Recently, IMDPs with continuous-action spaces have gained attention due to their role in finite-state abstraction of stochastic dynamical systems~\cite{ML-SBA-CB:12,SA-ML-LL:22} and in reachability analysis of stochastic systems~\cite{MD-JH-SC:21}.
Value iterations for IMDPs with continuous action-spaces are studied in~\cite{GD-ML-MM-LL:22} and in~\cite{SH-BM:18}.

One of the main challenges in studying IMDPs with continuous action-space arises in computing their optimal policies. 
It turns out that most of the existing iterative algorithms for estimating optimal policy of MDPs and IMDPs (including value iteration, policy-iteration, and their interval-valued counterparts) require solving an optimization problem in the action variables at each iteration step.
Unlike finite-state finite-action IMDPs where the optimization over the action-space can be implemented efficiently, for IMDPs with continuous action-spaces, optimization over action variables can lead to two important challenges.
First, in the absence of any structure for the optimization problem (\emph{e.g.} convexity/concavity of the cost function), it is generally necessary to resort to heuristic algorithms to approximate the solutions of these optimization problems. These heuristic methods can significantly degrade the quality of the estimated optimal policies and can ruin any guarantee on the optimality of the solutions. 
Secondly, in large-scale IMDPs, solving these optimization problems at each iteration step is computationally complicated, potentially leading to intractability of finding the optimal values.
Most of the existing literature on IMDPs focuses on %
discretizing the action-space and then using known results about discrete-action IMDPs. However, this approximation is sub-optimal and scales poorly with the dimension of the action-space~\cite{GD-ML-MM-LL:22}. The only exception is~\cite{SA-ML-LL:22} which provides a computationally efficient reformulation of interval value iterations.

\paragraph*{Contributions}

In this paper, we study IMDPs with continuous action-spaces from a dynamical system perspective. 
In particular, we use monotone system theory and contraction theory to study convergence of their value iterations for both pessimistic and optimistic policies, that is, policies under the assumption of a adversarial agent and a cooperative agent, respectively. 
By considering the value iterations in IMDPs as dynamical systems, we study contractivity and monotonicity of the pessimistic and optimistic value iterations with respect to the $\ell_{\infty}$-norm and the standard partial order.
Next, given an IMDP, we introduce another IMDP, called the action-space relaxation IMDP, obtained by bounding its probability transition and rewards using suitable convex/concave functions. 
As our first main result, we use a dynamical systems perspective to show that the optimal value of action-space relaxation IMDP provides bound on the optimal value of the original MDP. 
As our second main result, given an action-space relaxation IMDP, we propose to reduce the computational burden of the interval value iterations by implementing the policy optimization problem in an iteration-distributed fashion. We consider the value iteration and the policy optimization problem as an interconnected dynamical system and leverage the contractivity of the value iterations to provide guarantees for convergence of the interconnected system to the optimal value of the action-space relaxation IMDP. 

\section{Notations and Mathematical Preliminary}

For every $p\in [1,\infty]$, we denote the $\ell_p$-norm on $\real^n$ by $\|\cdot\|_{p}$. Given two sets $X$ and $Y$, the set of all the maps from $X$ to $Y$ is denoted by $Y^{X}$. For any compact set $X\subseteq \real^n$, we define $\|X\|_{\infty} = \setdef{\|x-y\|_{\infty}}{x,y\in X}$. Let $A,B\in \real^{n\times n}$, then we write $A\succeq B$ if $A-B$ is a positive semi-definite matrix. Let $S$ be a finite set with $n$ elements and let $v\in \real^S$. We define $\{1_v,\ldots, n_v\}$ as an ordered permutation of elements of the set $S$ such that $v(1_v)\ge v(2_v)\ge \ldots \ge v(n_v)$. The set of all compact interval subsets of $[a,b]$ is denoted by $\mathrm{Interval}_{[a,b]}$ and the set of all compact interval subsets of $\real$ is denoted by $\mathrm{Interval}_{\real}$, i.e., we have 
\begin{align*}
    \mathrm{Interval}_{[a,b]} &= \setdef{[x,y]}{a\le x\le y\le b}\\
    \mathrm{Interval}_{\real} &= \setdef{[x,y]}{x\le y}.
\end{align*}
Given an operator $f:\real^n\to \real^n$, we say that $f$ is \emph{monotone} if, for every $x\le y$, we have $f(x)\le f(y)$. Given an operator $f:\real^n\to \real^n$, we say that $f$ is \emph{monotone} if, for every $x\le y$, we have $f(x)\le f(y)$. Given a norm $\|\cdot\|$ on $\real^n$, we say that $f$ is \emph{contracting} with rate $\lambda\in (0,1)$ with respect to the norm $\|\cdot\|$, if
\begin{align*}
    \|f(x)-f(y)\| \le \lambda\|x-y\|,\quad \mbox{ for every } x,y\in \real^n.
\end{align*}
Given a compact set $\mathcal{X}\subseteq \real^n$, the orthogonal projection into $\mathcal{X}$ is denoted by $\mathrm{Proj}_{\mathcal{X}}$, i.e., $\mathrm{Proj}_{\mathcal{X}}(y) = \mathrm{argmin}_{x\in \mathcal{X}}\|x-y\|_2$, for every $y\in \real^n$.  
We also recall the setting of a discounted infinite-horizon Markov Decision Process (MDP) with continuous action-space. An MDP with continuous action-space is a tuple $\mathcal{M} = (S,A,P,R,\gamma)$ where
\begin{enumerate}
\item $S$ is a finite set of states.
\item $A\subseteq\real^m$ is a compact action space.
\item $P: S\times S\times A\to [0,1]$ is the transition
  probability function, i.e., for every $s\in S$ and every $a\in A$,
  $P(s',s,a)$ is the probability of arriving at state $s'$ by taking
  action $a$ in the state $s$. We assume $0\le
  P(s',s,a)\le 1$ for every $s,s'\in S$ and every $a\in A$ and we
  have $\sum_{s'\in S}  P(s',s,a) =1$. 
  \item $R: S\times A\to \real$ is the reward function where
    $R(s,a)$ is the cost of taking action $a$ at state $s$.
 \item  $\gamma\in (0,1)$ is a discount factor.  
 \end{enumerate}

A \emph{policy} for the MDP $\mathcal{M}$ is a vector $\pi\in A^{S}$ which assigns an action $a$ to each state $s$\footnote{A policy defined this way is usually referred to as a Markovian deterministic stationary policy in the literature~\cite{MLP:14}.}.
For every policy $\pi\in A^{S}$, we  define the value function $V^{\mathcal{M}}_{\pi}:S\to \real$ as %
\begin{align}\label{eq:valuefunction}
  V^{\mathcal{M}}_\pi(s) = \mathbb{E}\left(\sum_{t=0}^{\infty} \gamma^t R(s_t,\pi(s_t))\Big| s_0=s\right)
\end{align}
where $\{s_t\}_{t=0}^{\infty}$ is a time sequence of states starting from $s_0=s$ and following the policy $\pi$. The goal is to find a policy $\pi^*\in A^{S}$ which maximizes the value function $V^{\mathcal{M}}_{\pi}$, i.e., a policy $\pi^*\in A^{S}$ such that
\begin{align}\label{eq:optimal}
    \pi^* = \mathrm{argmax}_{\pi\in A^{S}}V^{\mathcal{M}}_{\pi}.
\end{align}
In general, it can be shown that the optimization problem~\eqref{eq:optimal} has a unique optimal value $V^*$, which is obtained at a policy $\pi^*\in A^{S}$~\cite[Theorem 6.1.1]{MLP:14}, i.e., $V^{\mathcal{M}}_{\pi^*}=V^*$. It can be shown that the optimal value $V^*$ satisfies
\begin{align*}
  V^*(s) = R(s,\pi^*(s)) + \gamma \sum_{s'\in S} P(s',s,\pi^*(s))V^*(s').
  \end{align*}
The \emph{Bellman-policy operator} $\OF:\real^{S}\times A^{S} \to \real^{S}$ is 
\begin{align}\label{eq:bellman-FP}
  \OF(v,\pi)(s) :=  R(s,\pi(s)) + \gamma \sum_{s'\in S}
  P(s',s,\pi(s))v(s')
  \end{align}
  and the \emph{Bellman operator} $\OG:\real^{S}\to \real^{S}$ is defined by 
  \begin{align}\label{eq:bellman}
      \OG(v)(s) &:= \max_{a\in A} \left\{R(s,a) + \gamma \sum_{s'\in S}
  P(s',s,a)v(s')\right\}.
  \end{align}
  Equivalently, using the vector notation, we have $\OG(v) = \max_{\pi\in A^{S}} \OF(v,\pi)$, for every $v\in \real^{S}$. It is known that the Bellman operator $\OG$ is contracting with respect to the $\ell_{\infty}$-norm, monotone with respect to the standard partial ordering, and the optimal value $V^*$ is the fixed point of the Bellman operator, i.e., $V^{*}=\OG(V^*)$~\cite[Theorem 6.2.3]{MLP:14}.

  \section{Interval Markov Decision Process}

  In this section, we introduce \emph{Interval Markov Decision Processes (IMDPs)} as a class of Markov models where the cost functions and probability transitions are unknown and belong to  suitable intervals. An IMDP is a tuple $\mathcal{IM} = (S,A,[P],[R],\gamma)$ where
\begin{enumerate}
\item $S$ is a finite set of states.
\item $A\subseteq \real^m$ is a compact action space.
\item $[P]\in S\times S\times A\to \mathrm{Interval}_{[0,1]}$ denotes the transition
  probability intervals, i.e., for every $s\in S$ and every $a\in A$,
  $[P](s',s,a) = [\underline{P}(s',s,a),\overline{P}(s',s,a)]$ is the probability interval of arriving to the state $s'$ by taking
  action $a$ in the state $s$. For the sake of consistency,
  we assume that $\sum_{s'\in S}  \underline{P}(s',s,a) \le 1\le \sum_{s'\in S}  \overline{P}(s',s,a)$.
  \item $[R]:S\times A\to \mathrm{Interval}_{\real}$ denotes the reward function where
    $[R](s,a) = [\underline{R}(s,a), \overline{R}(s,a)]$ is the reward interval of taking action $a$ at state $s$.
 \item  $\gamma\in (0,1)$ is a discount factor.  
 \end{enumerate}
 
 An MDP $\mathcal{M} = (S,A,P,R,\gamma)$ belongs to the IMDP $\mathcal{IM}=(S,A,[P],[R],\gamma)$, and we write $\mathcal{M}\in \mathcal{IM}$, if 
 \begin{align*}
     \underline{P}(s',s,a) &\le P(s',s,a)\le \overline{P}(s',s,a), \\
     \underline{R}(s,a) &\le R(s,a)\le \overline{R}(s,a),
 \end{align*}
 for every $s,s'\in S$ and every $a\in A$. A policy for $\mathcal{IM}$ is a vector $\pi\in A^{S}$ that assigns an action $a$ to each state $s$. %

   \begin{remark}\textit{(Comparison with the literature)} Our definition of IMDPs generalizes the classical definitions in~\cite{RG-SL-TD:00,AN-LEG:05} which assume finite action-spaces. This generalization is motivated by, \emph{e.g.,} applications in robotics~\cite{JJ-YZ-SC:22} and abstraction of stochastic dynamical systems~\cite{MD-JH-SC:21,GD-ML-MM-LL:22}. 
     Moreover, two different interpretations for IMDPs have been proposed in the literature. The first interpretation considers an IMDP as an MDP with uncertain parameters~\cite{RG-SL-TD:00,AN-LEG:05}, whereas the second interpretation considers an IMDP as an MDP interacting with an agent who resolves uncertain probabilities~\cite{SHQL-AA-PLG-BA:21,GD-ML-MM-LL:22}. 
 \end{remark}
\smallskip

Given an IMDP $\mathcal{IM}$ and a policy $\pi\in A^{S}$, we study the possible ranges of the value function~\eqref{eq:valuefunction} for every $\mathcal{M}\in \mathcal{IM}$. First, for every $(s,a)\in S\times A$, we  define $\Delta^{\mathcal{IM}}_{s,a}$ as the set of all $p:S\to [0,1]$ such that,
 \begin{align} \label{eq:deltasa}
    \underline{P}(s',s,a) &\le p(s') \le \overline{P}(s',s,a), \mbox{ for all }s'\in S\nonumber\\
    &\sum\nolimits_{r\in S} p(r) = 1. 
 \end{align}
Using the set $\Delta^{\mathcal{IM}}_{s,a}$, we define the \textit{interval Bellman-policy operator} for $\mathcal{IM}$ as the map $\left[\begin{smallmatrix} \underline{\OF} \\ \overline{\OF} \end{smallmatrix}\right]: \real^{S}\times A^{S}\to \real^{S}\times \real^S$:
 \begin{align}\label{eq:IB}
     \underline{\OF}(v,\pi)(s) &= \underline{R}(s,\pi(s)) + \gamma\min_{p\in \Delta^{\mathcal{IM}}_{s,\pi(s)}} \sum_{s'\in S}p(s')v(s'),\nonumber\\
     \overline{\OF}(v,\pi)(s) &= \overline{R}(s,\pi(s)) + \gamma\max_{p\in \Delta^{\mathcal{IM}}_{s,\pi(s)}} \sum_{s'\in S}p(s')v(s'),
 \end{align}
 and the \textit{interval Bellman operator} for $\mathcal{IM}$ as the map $ \left[\begin{smallmatrix} \underline{\OG} \\ \overline{\OG} \end{smallmatrix}\right]: \real^{S}\to \real^{S}\times \real^S$:
 \begin{align}\label{eq:IB-2}
     \underline{\OG}(v)(s) &= \max_{a\in A}\left\{\underline{R}(s,a) + \gamma\min_{p\in \Delta^{\mathcal{IM}}_{s,a}} \sum_{s'\in S}p(s')v(s')\right\},\nonumber\\
     \overline{\OG}(v)(s) &= \max_{a\in A}\left\{\overline{R}(s,a) + \gamma\max_{p\in \Delta^{\mathcal{IM}}_{s,a}} \sum_{s'\in S}p(s')v(s')\right\}.
 \end{align}
 Equivalently, using the vector notation,  we have $\underline{\OG}(v) = \max_{\pi\in A^{S}}\underline{\OF}(v,\pi)$ and $\overline{\OG}(v) =\max_{\pi\in A^{S}}\overline{\OF}(v,\pi)$. 
 for every $v\in \real^{S}$. In the next theorem, we show that the interval Bellman (resp. Bellman-policy) operator is monotone and contracting with respect to the $\ell_{\infty}$-norm and can be used to provide upper and lower bounds on the Bellman (resp. Bellman-policy) operator of every MDP that belongs to $\mathcal{IM}$. 
 
  \begin{theorem}[Interval Bellman operator]\label{thm:IB}
     Consider an IMDP $\mathcal{IM} = (S,A,[P],[R],\gamma)$ with the interval Bellman-policy operator and the interval Bellman operator in~\eqref{eq:IB} and~\eqref{eq:IB-2}, respectively. Let $\mathcal{M}$ be an MDP such that $\mathcal{M}\in \mathcal{IM}$ with the Bellman-policy operator and the Bellman operator operator in~\eqref{eq:bellman-FP} and~\eqref{eq:bellman}, respectively. Then,
     \begin{enumerate}
         \item\label{p1:IB} for every $\pi\in A^{S}$, the operators $v\mapsto \underline{\OF}(v,\pi)$ and $v\mapsto \overline{\OF}(v,\pi)$ are monotone and contracting with respect to the $\ell_{\infty}$-norm with rate $\gamma$ and
         \begin{align}\label{eq:1}
         \underline{\OF}(v,\pi) \le \OF(v,\pi) \le \overline{\OF}(v,\pi),\quad\mbox{for all }v\in \real^S.
     \end{align}
         \item\label{p2:IB} the operators $\underline{\OG}$ and $\overline{\OG}$ are monotone and contracting with respect to the $\ell_{\infty}$-norm with rate $\gamma$ and 
         \begin{align}\label{eq:2}
         \underline{\OG}(v) \le \OG(v) \le \overline{\OG}(v),\quad\mbox{for all }v\in \real^S.
     \end{align}
     \end{enumerate}
   \end{theorem}
   \smallskip
\begin{proof}
    Regarding part~\ref{p1:IB}, consider $v,w\in \real^S$ such that
    $v\le w$. As a result, we have $\sum_{s'\in S}p(s')v(s')\le \sum_{s'\in S}p(s')w(s')$, for every $p\in \Delta^{\mathcal{IM}}_{s,a}$ and every $s,a\in S\times A$. Therefore, for every $s\in S$  and every $\pi\in A^S$,
\begin{multline*}
    \underline{\OF}(v,\pi)(s) = \underline{R}(s,\pi(s)) + \min_{p\in \Delta^{\mathcal{IM}}_{s,\pi(s)}}\sum_{s'\in S}p(s')v(s') \\ \le  \underline{R}(s,\pi(s)) + \min_{p\in \Delta^{\mathcal{IM}}_{s,\pi(s)}}\sum_{s'\in S}p(s')w(s')= \underline{\OF}(w,\pi)(s)
\end{multline*}
This implies that the operator $v\mapsto \underline{\OF}(v,a)$ is monotone.  Moreover, for every $\pi\in A^{S}$, 
\begin{multline*}
    \|\underline{\OF}(v,\pi)-\underline{\OF}(w,\pi)\|_{\infty} \\ = \gamma \max_{s\in S}\left|\min_{p\in \Delta^{\mathcal{IM}}_{s,\pi(s)}}\sum_{s'\in S}p(s')v(s')- \min_{p\in \Delta^{\mathcal{IM}}_{s,\pi(s)}}\sum_{s'\in S}p(s')w(s')\right| 
\end{multline*}
Let $s\in S$ and note that
\begin{align*}
  \min_{p\in \Delta^{\mathcal{IM}}_{s,\pi(s)}}\sum_{s'\in S}p(s')v(s')\le \min_{p\in \Delta^{\mathcal{IM}}_{s,\pi(s)}}\sum_{s'\in S}p(s')w(s')
\end{align*}
 Since $\Delta^{\mathcal{IM}}_{s,\pi(s)}$ is compact, there exists $p^*\in \Delta^{\mathcal{IM}}_{s,\pi(s)}$ such that $\min_{p\in \Delta^{\mathcal{IM}}_{s,\pi(s)}}\sum_{s'\in S}p(s')v(s')=\sum_{s'\in S}p^*(s')v(s')$. This implies that, for every $s\in S$, 
\begin{multline*}
  \left|\min_{p\in \Delta^{\mathcal{IM}}_{s,\pi(s)}}\sum_{s'\in
      S}p(s')v(s') - \min_{p\in
      \Delta^{\mathcal{IM}}_{s,\pi(s)}}\sum_{s'\in S}p(s')w(s')\right|
  \\  =   \min_{p\in \Delta^{\mathcal{IM}}_{s,\pi(s)}}\sum_{s'\in
    S}p(s')w(s') - \min_{p\in
    \Delta^{\mathcal{IM}}_{s,\pi(s)}}\sum_{s'\in S}p(s')v(s')\\ \quad
  = \min_{p\in \Delta^{\mathcal{IM}}_{s,\pi(s)}}\sum_{s'\in
    S}p(s')w(s') - \sum_{s'\in S}p^*(s')v(s')  \\ \le  \sum_{s'\in S}p(s')(v(s')-w(s')) \le \|v-w\|_{\infty}.
\end{multline*}
As a result, we get $\|\underline{\OF}(v,\pi)-\underline{\OF}(v,\pi)\|_{\infty}\le \gamma \|v-w\|_{\infty}$. This means that $v\mapsto \underline{\OF}(v,\pi)$ is contracting with respect to the $\ell_{\infty}$-norm with the rate $\gamma$. Similarly, one can show that $v\mapsto \overline{\OF}(v,\pi)$ is monotone and contracting with respect to the $\ell_{\infty}$-norm with the rate $\gamma$. To show the inequalities in~\eqref{eq:1}, note that, for every $a\in A$ and every $s,s'\in S$, we have $\underline{P}(s',s,a) \le P(s',s,a)\le \overline{P}(s',s,a)$. Using the definition of $\Delta^{\mathcal{IM}}_{s,a}$ in equation~\eqref{eq:deltasa}, we get that
 \begin{multline*}
     \min_{p\in \Delta^{\mathcal{IM}}_{s,a}} \sum_{s'\in
       S}p(s')v(s')\le \sum_{s'\in S}P(s',s,a)v(s') \\ \le \max_{p\in \Delta^{\mathcal{IM}}_{s,a}} \sum_{s'\in S}p(s')v(s'),
 \end{multline*}
 for every $s\in S$ and every $a\in A$. Thus, for every $v\in \real^S$ and $\pi\in A^S$, we have 
 \begin{align*}
  \underline{\OF}&(v,\pi)(s)=\underline{R}(s,\pi(s)) + \gamma \min_{p\in \Delta^{\mathcal{IM}}_{s,\pi(s)}} \sum_{s'\in S}p(s')v(s')\\ & \le  R(s,\pi(s)) + \gamma     \sum_{s'\in S}P(s',s,\pi(s))v(s') = \OF(v,\pi)(s) \\ & \le \overline{R}(s,\pi(s)) + \gamma \max_{p\in \Delta^{\mathcal{IM}}_{s,\pi(s)}} \sum_{s'\in S}p(s')v(s') = \overline{\OF}(v,\pi)(s). 
 \end{align*}

Regarding part~\ref{p2:IB}, consider $v,w \in \real^S$ such that $v\le w$. Then, using part~\ref{p1:IB}, we have $\underline{\OF}(v,\pi)\le \underline{\OF}(w,\pi)$, for every $\pi\in A^S$. This implies that 
\begin{align*}
    \underline{\OG}(v) = \max_{\pi\in A^S} \underline{\OF}(v,\pi)\le \max_{\pi\in A^S}\underline{\OF}(w,\pi) = \underline{\OG}(w). 
\end{align*}
This means that $\underline{\OG}$ is monotone. On the other hand, let $\pi^*\in A^S$ be such that $\OF(v,\pi^*) = \max_{\pi\in A^S}\underline{\OF}(v,\pi)$. Therefore, we have
\begin{align*}
    \|\underline{\OG}(v)-\underline{\OG}(w)\|_{\infty} &= \|\max_{\pi\in A^S}\underline{\OF}(v,\pi)-\max_{\pi\in A^S}\underline{\OF}(w,\pi)\|_{\infty} \\ & = \max_{\pi\in A^S}\underline{\OF}(w,\pi) - \max_{\pi\in A^S}\underline{\OF}(v,\pi) \\ & 
    =\max_{\pi\in A^S}\underline{\OF}(w,\pi) -\underline{\OF}(v,\pi^*) \\ & \le \underline{\OF}(w,\pi^*)-\underline{\OF}(v,\pi^*) \le \gamma\|v-w\|_{\infty},
\end{align*}
where the last inequality holds by contractivity of $v\mapsto \underline{\OF}(v,\pi^*)$ proved in part~\ref{p1:IB}. Thus, $\underline{\OG}$ is contracting with respect to the $\ell_{\infty}$-norm with the rate $\gamma$. Similarly, one can show that $\overline{\OG}$ is monotone and contracting with respect to the $\ell_{\infty}$-norm with the rate $\gamma$. Finally the inequalities~\eqref{eq:2} follows from part~\ref{p1:IB} and definition of $\underline{\OG}$ and $\overline{\OG}$. \end{proof}
     
 Computing the interval Bellman operator using the equation~\eqref{eq:IB} requires solving two linear programs in $p$ and can be computationally intractable for large-scale IMDPs. We first introduce two useful notations. Consider $\mathcal{IM}=(S,A,[P],[R],\gamma)$ with $(s,a)\in S\times A$ and $v\in \real^{S}$. Then we define $\underline{\iota}(v,s,a)$ as the largest integer $j\in \{1,\ldots,n\}$ satisfying
 \begin{multline*}
 \underline{P}(j_v,s,a)\le 1- \sum_{i=1}^{j-1}\underline{P}(i_v,s,a)-\sum_{i=j+1}^{n}\overline{P}(i_v,s,a)\\ \le \overline{P}(j_v,s,a),
 \end{multline*}
 and $\overline{\iota}(v,s,a)$ as the largest integer $k\in \{1,\ldots,n\}$ satisfying
  \begin{multline*}
    \underline{P}(k_v,s,a)\le 1- \sum_{i=1}^{k-1}\overline{P}(i_v,s,a)-\sum_{i=k+1}^{n}\underline{P}(i_v,s,a)\\ \le \overline{P}(k_v,s,a).
 \end{multline*}
 Note that existence of $\underline{\iota}(v,s,a)$ and $\overline{\iota}(v,s,a)$ follows from the inequality $\sum_{s'\in S}  \underline{P}(s',s,a) \le 1\le \sum_{s'\in S}  \overline{P}(s',s,a)$. We also define the operators $\Omega^{\mathcal{IM}}:\real^{S}\times S\times A\to \real$ and $\Lambda^{\mathcal{IM}}:\real^{S}\times S\times A\to \real$ as follows: 
   \begin{align}\label{eq:omega}
       &\Omega^{\mathcal{IM}}(v,s,a)=
       \sum_{i=1}^{j} (v(i_v)-v(j_v))\underline{P}(i_v,s,a) \nonumber\\ & +  \sum_{i=j}^{n} (v(i_v)-v(j_v))\overline{P}(i_v,s,a) + v(j_v), \nonumber\\
       &\Lambda^{\mathcal{IM}}(v,s,a)=\sum_{i=1}^{k} (v(i_v)-v(k_v))\overline{P}(i_v,s,a) \nonumber\\ & +  \sum_{i=k}^{n} (v(i_v)-v(k_v))\underline{P}(i_v,s,a) + v(k_v),
   \end{align}
   where $j=\underline{\iota}(v,s,a)$ and $k=\overline{\iota}(v,s,a)$. The next proposition provides a closed-from expression for the interval Bellman-policy using the lower and upper probability transition bounds.    
 \smallskip
 \begin{proposition}[Bellman-policy operator]\label{thm:closedform}
 Consider the IMDP $\mathcal{IM}=(S,A,[P],[R],\gamma)$ with the interval Bellman-policy operator $\left[\begin{smallmatrix} \underline{\OF} \\ \overline{\OF} \end{smallmatrix}\right]$ defined in~\eqref{eq:IB}. Then 
   \begin{align*}
       \underline{\OF}(v,\pi)(s) &= \underline{R}(s,\pi(s)) +\gamma \Omega^{\mathcal{IM}}(v,s,\pi(s)),\\
     \overline{\OF}(v,\pi)(s) &= \overline{R}(s,\pi(s)) + \gamma \Lambda^{\mathcal{IM}}(v,s,\pi(s)),
   \end{align*}
   where $\Omega^{\mathcal{IM}}$ and $\Lambda^{\mathcal{IM}}$ are defined in~\eqref{eq:omega}.
 \end{proposition}
 \begin{proof}
 We start by showing that $\min_{p\in
   \Delta^{\mathcal{IM}}_{s,a}}\sum_{s'\in S}p(s')v(s') = \Omega^{\mathcal{IM}}(v,s,a)$. In the course of this proof we set $j=\underline{\iota}(v,s,a)$. Note that, we can compute
\begin{align*}
    \sum_{s'\in S}p(s')v(s') & = \sum_{i=1}^{j} (v(i_v)-v(j_v))p(i_v) \\ & + \sum_{i=j}^{n} (v(i_v)-v(j_v))p(i_v) + v(j_v),
\end{align*}
where the equality holds using the fact that $\sum_{s'\in S}p(s') =1$. The above equality implies that 
\begin{align*}
    \sum_{s'\in S}p(s')v(s') & = \sum_{i=1}^{j} (v(i_v)-v(j_v))p(i_v) \\ & \qquad + \sum_{i=j}^{n} (v(i_v)-v(j_v))p(i_v) + v(j_v) \\ & \ge \sum_{i=1}^{j} (v(i_v)-v(j_v))\overline{P}(i_v,s,a) \\ & \qquad + \sum_{i=j}^{n} (v(i_v)-v(j_v))\underline{P}(i_v,s,a) + v(j_v) \\ & \ge \Omega^{\mathcal{IM}}(v,s,a), 
\end{align*}
where the second inequality holds because $\underline{P}(i_v,s,a)\le p(i_v)\le \overline{P}(i_v,s,a)$ and because $v(i_v)-v(j_v) \ge 0$ if $i\le j$ and $v(i_v)-v(j_v) \le 0$ if $i>j$. This means that, for every $p\in \Delta^{\mathcal{IM}}_{s,a}$, we have $\sum_{s'\in S}p(s')v(s')\ge \Omega^{\mathcal{IM}}(v,s,a)$. Now we show that $\min_{p\in \Delta^{\mathcal{IM}}_{s,a}} \sum_{s'\in S}p(s')v(s') =  \Omega^{\mathcal{IM}}(v,s,a)$. Let $v\in \real^n$ and define $p^*\in \real^{S}$ as follows:
\begin{align*}
    p^*(i_v) = \begin{cases}
    \overline{P}(i_v,s,a) & i\in \{1,\ldots,j-1\}\\
    \xi & i = j\\
    \underline{P}(i_v,s,a) & i\in \{j+1,\ldots,n\}.
    \end{cases}
\end{align*}
where $\xi = 1-\sum_{i=1}^{j-1}\overline{P}(i_v,s,a)
-\sum_{i=j+1}^{n}\underline{P}(i_v,s,a)$. Note that by definition of
$\underline{\iota}(v,s,a)$, we have $p^*\le \vect{1}_{|S|}$. Moreover
$\sum_{s'\in S}p^*(s')=1$ and $\underline{P}(s',s,a)\le p^*(s')\le \overline{P}(s',s,a)$, for every $s'\in S$. Therefore $p^*\in \Delta^{\mathcal{IM}}_{s,a}$ and with this choice of $p^*$, we have $\sum_{s'\in S}p^*(s')v(s') = \Omega^{\mathcal{IM}}(v,s,a)$. The proof of $\max_{p\in \Delta^{\mathcal{IM}}_{s,a}} \sum_{s'\in S}p(s')v(s')= \Lambda^{\mathcal{IM}}(v,s,a)$ is similar and we omit it for the sake of brevity. \end{proof}
\begin{remark}[Comparison with the literature]
For finite action space $A$, Theorem~\ref{thm:closedform} recovers the results in~\cite{RG-SL-TD:00} and the pseudocode in~\cite[Figure 8]{RG-SL-TD:00}. For infinite-dimensional action space $A$, Theorem~\ref{thm:closedform} provides a simpler form compared to~\cite[Theorem 3.2]{GD-ML-MM-LL:22}.
\end{remark}
 
 \smallskip

    It is known that the notion of \emph{optimal policy}, as defined in~\eqref{eq:optimal} for MDPs, is not well-defined for IMDPs~\cite{RG-SL-TD:00}. This is due to the fact that the value functions of IMDPs are interval-valued and the set of intervals do not have a standard partial order.  However, given an IMDP $\mathcal{IM}=(S,A,[P],[R],\gamma)$, one can define two policies, namely  the \textit{pessimistic optimal policy} and the \textit{optimistic optimal policy}, which provide certain type of optimally for the value function. The pessimistic optimal policy $\pi^*_{\mathrm{p}}\in A^{S}$ is the unique policy defined by
       \begin{align*}
          \pi^*_{\mathrm{p}} = \mathrm{argmax}_{\pi\in A^{S}} \left(\min_{\mathcal{M}\in \mathcal{IM}} V^{\mathcal{M}}_{\pi}\right), 
       \end{align*}
       and the pessimistic value function is given by $V^*_{\mathrm{p}}=\min_{\mathcal{M}\in \mathcal{IM}} V^{\mathcal{M}}_{\pi^*_{\mathrm{p}}}$. The optimistic optimal policy $\pi^*_{\mathrm{o}}\in A^{S}$ is the unique policy defined by
       \begin{align*}
          \pi^*_{\mathrm{o}} = \mathrm{argmax}_{\pi\in A^{S}} \left(\max_{\mathcal{M}\in \mathcal{IM}} V^{\mathcal{M}}_{\pi} \right), 
       \end{align*}
       and the optimistic value function is given by $V^*_{\mathrm{o}}=\max_{\mathcal{M}\in \mathcal{IM}} V^{\mathcal{M}}_{\pi^*_{\mathrm{o}}}$. From a game-theoretic perspective, the pessimistic optimal policy can be considered as the optimal policy of the IMDP $\mathcal{IM}$ in presence of a competitive adversary who resolves uncertain probabilities, and the optimistic optimal policy can be considered as the optimal policy of the IMDP $\mathcal{IM}$ the presence of a cooperative agent who resolves uncertain probabilities. Given an IMDP $\mathcal{IM}$, we define the \emph{pessimistic value iteration} by
        \begin{align}\label{eq:valueBE-pes}
        v^{k+1} & = \underline{\OG}(v^k)= \max_{\pi\in A^{S}} \underline{\OF}(v^k,\pi),
       \end{align}
       and we define the \emph{optimistic value iteration} by
       \begin{align}\label{eq:valueBE-opt}
       v^{k+1}= \overline{\OG}(v^k)= \max_{\pi\in A^{S}} \overline{\OF}(v^k,\pi).
       \end{align}
       where $\left[\begin{smallmatrix}\underline{\OG}\\\overline{\OG} \end{smallmatrix}\right]$ and $\left[\begin{smallmatrix}\underline{\OF}\\\overline{\OF} \end{smallmatrix}\right]$ are the interval Bellman operator and the interval Bellman-policy operator of $\mathcal{IM}$, respectively.  
    The next theorem establishes that the pessimistic and optimistic value iterations~\eqref{eq:valueBE-pes} and~\eqref{eq:valueBE-opt} can be used to compute the pessimistic and optimistic optimal policies of IMDPs.

 \begin{theorem}[Value iterations as dynamical systems]\label{thm:pasopt}
 Consider the IMDP $\mathcal{IM}=(S,A,[P],[R], \gamma)$ with the pessimistic and optimistic policies $\pi^*_{\mathrm{p}},\pi^*_{\mathrm{o}}\in A^{S}$ with the interval Bellman-policy operator~\eqref{eq:IB} and the interval Bellman operator~\eqref{eq:IB-2}. Then, the following statements hold:
 \begin{enumerate}
     \item\label{p2:pasopt} the pessimistic value iteration~\eqref{eq:valueBE-pes} is a monotone contracting dynamical system with the unique globally exponentially stable equilibrium point $V^*_{\mathrm{p}}$ and the pessimistic optimal policy $\pi^*_{\mathrm{p}}$ is obtained by $\pi^*_{\mathrm{p}}  = \mathrm{argmax}_{\pi\in A^{S}}\underline{\OF}(V^{*}_{\mathrm{p}},\pi)$. 
     \item\label{p3:pasopt} the optimistic value iteration~\eqref{eq:valueBE-opt} is a monotone  contracting dynamical system with the unique globally exponentially stable equilibrium point $V^*_{\mathrm{o}}$ and the optimistic optimal policy $\pi^*_{\mathrm{o}}$ is obtained by $\pi^*_{\mathrm{o}}  = \mathrm{argmax}_{\pi\in A^{S}}\overline{\OF}(V^{*}_{\mathrm{o}},\pi)$.
 \end{enumerate}
 \end{theorem}
 \smallskip
 \begin{proof}
     Regarding part~\ref{p2:pasopt}, by Theorem~\ref{thm:IB}\ref{p2:IB}, the interval Bellman operator $\underline{\OG}$ is monotone and contracting with rate $\gamma$ with respect to the norm $\|\cdot\|_{\infty}$. This implies that the discrete-time dynamical system~\eqref{eq:valueBE-pes} is contracting with rate $\gamma$ with respect to the norm $\|\cdot\|_{\infty}$. As a result, it has a unique equilibrium point $V^{*}_{\mathrm{p}}$ which is globally exponentially stable. Note that, by the definition of interval Bellman operator $\underline{\OG}$, we have $\underline{\OG}(V^{*}_{\mathrm{p}})(s) = \max_{a\in A}\underline{\OF}(V^{*}_{\mathrm{p}}(s),a)$. Thus, the pessimistic optimal policy satisfies $\pi_{\mathrm{p}}^*(s) = \mathrm{argmax}_{a\in A}\underline{\OF}(V^{*}_{\mathrm{p}}(s),a)$, for every $s\in S$. This completes the proof of part~\ref{p2:pasopt}. Regarding part~\ref{p3:pasopt}, the proof follows from a similar argument as in the proof of part~\ref{p2:pasopt}. 
 \end{proof}
 \smallskip
 \begin{remark}[A dynamical system perspective]
 The fact that pessimistic (resp. optimistic) value iterations is contracting and the pessimistic (resp. optimistic) optimal policy is its fixed points is known in the literature~\cite[Theorems 10,11,12]{RG-SL-TD:00}. However, Theorem~\ref{thm:pasopt} provides a discrete-time dynamical system perspective to the pessimistic (resp. optimistic) value iterations~\eqref{eq:valueBE-pes} (resp. equation~\eqref{eq:valueBE-opt}) and highlights their less-studied property of monotonicity.
 \end{remark}

\section{Efficient estimation of optimal policies}

Theorem~\ref{thm:pasopt} provides iterative algorithms for computing
the optimal policies in IMDPs. It turns out that implementing the pessimistic value iterations~\eqref{eq:valueBE-pes}
 (resp. optimistic value iterations~\eqref{eq:valueBE-opt}) requires solving the following $|S|$ nonlinear optimization problems at each iteration step:
 \begin{align}\label{eq:opt}
     \underline{\pi}^k = \mathrm{argmax}_{\pi\in A^{S}} \underline{\OF}(v^k,\pi)
 \end{align}
 (resp. $\overline{\pi}^{k} = \mathrm{argmax}_{\pi\in A^{S}}\overline{\OF}(v^k,\pi)$). 
This can cause two main challenges for computing the optimal policies:
\begin{enumerate}
    \item in the absence of any structure for the optimization problems~\eqref{eq:opt}, one needs to resort to heuristic algorithms to approximate the optimal solutions of~\eqref{eq:opt}. These heuristic algorithms can introduce sizable error in estimating the optimization problem and can significantly degrade the performance of the value iterations.
    \item  even when the optimization problems~\eqref{eq:opt} is convex, it is still necessary to solve $|S|$ optimization problems with $m$ variables at each iterations of the value iterations. 
    Thus, it is computationally challenging to implement the interval value iterations for large-scale IMDPs.
\end{enumerate}
In order to address the above mentioned challenges, we study IMDPs through the lens of dynamical systems. In the rest of this section, we focus on the pessimistic value iterations and pessimistic optimal policies. A parallel framework can be developed for optimistic value iterations and optimistic optimal policies but we omit it for the sake of brevity.

\subsection{Action-space relaxation IMDP} \label{sub:relax}

In this subsection, we introduce a relaxation of a given IMDP in its action variables by providing suitable bounds on its reward functions and its probability transition functions.

 \begin{definition}[Action-space pessimistic relaxation]\label{def:concaveimdp}
 Consider an IMDP $\mathcal{IM} = (S,A,[P],[R],\gamma)$. An \emph{action-space pessimistic relaxation} of $\mathcal{IM}$ is an IMDP $\mathcal{IM}^{\mathrm{cv}}=(S,A^{\mathrm{cv}}, [P^{\mathrm{cv}}],[R^{\mathrm{cv}}],\gamma)$ such that
 \begin{enumerate}
     \item $A\subseteq A^{\mathrm{cv}}\subseteq \real^m$ and $A^{\mathrm{cv}}$ is convex and compact in $\real^m$,
     \item for every $s',s\in S$ and every $a\in A$.
\begin{align*}
    &\underline{P}(s',s,a)\le \underline{P}^{\mathrm{cv}}(s',s,a),\quad 
    \overline{P}^{\mathrm{cv}}(s',s,a)\le \overline{P}(s',s,a),\\
    &\qquad\qquad\qquad\underline{R}(s,a)\le \underline{R}^{\mathrm{cv}}(s,a),
\end{align*}
     \item for every $s\in S$, $a \mapsto \underline{R}^{\mathrm{cv}}(s,a)$ is concave on $A^{\mathrm{cv}}$,
     \item for every $s',s\in S$, $a \mapsto \overline{P}^{\mathrm{cv}}(s',s,a)$ is convex and $a \mapsto \underline{P}^{\mathrm{cv}}(s',s,a)$ is  concave.
 \end{enumerate}

 \end{definition}
\smallskip

Given an action-space pessimistic relaxation $\mathcal{IM}^{\mathrm{cv}}$ for $\mathcal{IM}$, one can define its associated interval Bellman-policy operator $\left[\begin{smallmatrix} \underline{\OF}^{\mathrm{cv}} \\ \overline{\OF}^{\mathrm{cv}} \end{smallmatrix}\right]: \real^{S}\times (A^{\mathrm{cv}})^{S}\to \real^{S}\times \real^{S}$ and the interval Bellman operator $\left[\begin{smallmatrix} \underline{\OG}^{\mathrm{cv}} \\ \overline{\OG}^{\mathrm{cv}} \end{smallmatrix}\right]: \real^{S}\times (A^{\mathrm{cv}})^{S}\to \real^{S}\times \real^{S}$ as in equations~\eqref{eq:IB} and~\eqref{eq:IB-2}, respectively. Then, the pessimist value iterations for $\mathcal{IM}^{\mathrm{cv}}$ is given by 
        \begin{align}\label{eq:valueBE-pes-concave}
        &v^{k+1} = \underline{\OF}^{\mathrm{cv}}(v^{k},\pi^k),\nonumber\\
        &\pi^k = \mathrm{argmax}_{\pi\in (A^{\mathrm{cv}})^S} \underline{\OF}^{\mathrm{cv}}(v^{k},\pi),
       \end{align}
       and the pessimistic optimal value and the pessimistic optimal policy of $\mathcal{IM}^{\mathrm{cv}}$ is denoted by $V^{\mathrm{cv},*}_{\mathrm{p}}$ and $\pi^{\mathrm{cv},*}_{\mathrm{p}}$, respectively. 
    Given an IMDP $\mathcal{IM}$, the next theorem shows that the Bellman operator of the action-space pessimistic relaxation $\mathcal{IM}^{\mathrm{cv}}$ is an upper bound for the Bellman operator of $\mathcal{IM}$. 
    Using the classical comparison theorem for the pessimistic value iterations~\eqref{eq:valueBE-pes-concave}, it can be shown that the pessimistic optimal value of $\mathcal{IM}^{\mathrm{cv}}$ is an upper bound for the pessimistic optimal value of $\mathcal{IM}$.

     \begin{theorem}[Bellman operator of pessimistic relaxation]\label{thm:concave}
     Consider the IMDP $\mathcal{IM} = (S,A,[P],[R],\gamma)$ with an associated action-space pessimistic relaxation IMDP $\mathcal{IM}^{\mathrm{cv}} = (S,A^{\mathrm{cv}},[P^{\mathrm{cv}}],[R^{\mathrm{cv}}],\gamma)$. Then,
\begin{enumerate}
    \item\label{p1:concave} for every $v\in \real^{S}$ and $\pi \in A^{S}$, $\underline{\OF}(v,\pi)\le \underline{\OF}^{\mathrm{cv}}(v,\pi)$,
     \item\label{p2.5:concave} for every $s\in S$ and every $v\in \real^{S}$, $\pi\to \underline{\OF}^{\mathrm{cv}}(v,\pi)(s)$,
    is a concave function on $A^{\mathrm{cv}}$. 
    \item\label{p2:concave} for every $v\in \real^{S}$, $\underline{\OG}(v)\le \underline{\OG}^{\mathrm{cv}}(v)$. 
    \item\label{p3:concave}  we have $V^*_{\mathrm{p}}\le V^{\mathrm{cv},*}_{\mathrm{p}}$.
\end{enumerate}
     \end{theorem}
     \smallskip
   \begin{proof}  
   Regarding part~\ref{p1:concave}, recall the definition of $\Omega^{\mathcal{IM}^{\mathrm{cv}}}(v,s,a)$ in equation~\eqref{eq:omega} for the IMDP $\mathcal{IM}^{\mathrm{cv}}$. Then, 
   \begin{align*}
\Omega^{\mathcal{IM}^{\mathrm{cv}}}(v,s,\pi(s))&= \sum_{i=1}^{j} (v(i_v)-v(j_v))\underline{P}^{\mathrm{cv}}(i_v,s,a) \\ & +  \sum_{i=j}^{n} (v(i_v)-v(j_v))\overline{P}^{\mathrm{cv}}(i_v,s,a) + v(j_v) \\ & \ge  \sum_{i=1}^{j} (v(i_v)-v(j_v))\underline{P}(i_v,s,a) \\ & +  \sum_{i=j}^{n} (v(i_v)-v(j_v))\overline{P}(i_v,s,a) + v(j_v), 
   \end{align*}
   where $j=\underline{\iota}^{\mathrm{cv}}(v,s,a)$ is as defined in~\eqref{eq:omega}. Note that the first inequality above holds because, for every $i\in \{1,\ldots,j\}$, we have $v(i_v)-v(j_v)\ge 0$ and $\underline{P}^{\mathrm{cv}}(i_v,s,a)\ge \underline{P}(i_v,s,a)$ and, for every $i\in \{j,\ldots,n\}$, we have $v(i_v)-v(j_v)\le 0$ and $\overline{P}^{\mathrm{cv}}(i_v,s,a)\le \overline{P}(i_v,s,a)$. Given $\pi\in A^S$, we define $p^*:S \to \real$ by  
   \begin{align*}
    p^*(i_v) = \begin{cases}
    \underline{P}(i_v,s,\pi(s)) & i\in \{1,\ldots,j-1\}\\
    \xi & i = j\\
    \overline{P}(i_v,s,\pi(s)) & i\in \{j+1,\ldots,n\}.
    \end{cases}
\end{align*}
where $\xi = 1-\sum_{i=1}^{j-1}\underline{P}(i_v,s,\pi(s)) -\sum_{i=j+1}^{n}\overline{P}(i_v,s,\pi(s))$. It is easy to check $p^*\in \Delta^{\mathcal{IM}}_{s,\pi(s)}$ and
\begin{align*}
\Omega^{{\mathcal{IM}}^{\mathrm{cv}}}(v,s,\pi(s))\ge \sum_{s'\in S}p^*(s')v(s').
   \end{align*}
   As a result, using Proposition~\ref{thm:closedform}, we get
    \begin{align*}
      \underline{\OF}^{\mathrm{cv}}&(v,\pi)(s) = \underline{R}^{\mathrm{cv}}(s,\pi(s)) + \gamma \Omega^{\mathcal{IM}^{\mathrm{cv}}}(v,s,\pi(s)) \\ & \ge \underline{R}(s,\pi(s)) + \gamma \min_{\pi\in \Delta^{\mathcal{IM}}_{s,\pi(s)}} \sum_{s'\in S}p(s')v(s') =  \underline{\OF}(v,\pi)(s).
   \end{align*}
   where the last equality holds by the definition of $\underline{\OF}$.  
   Regarding part~\ref{p2.5:concave}, first note that, for every $i\in \{1,\ldots,j\}$, we have $v(i_v)-v(j_v)\ge 0$ and $a\mapsto \underline{P}^{\mathrm{cv}}(i_v,s,a)$ is concave and, for every $i\in \{j,\ldots,n\}$, we have $v(i_v)-v(j_v)\le 0$ and $a\mapsto \overline{P}^{\mathrm{cv}}(i_v,s,a)$ is convex. This implies that $\pi\mapsto \Omega^{\mathcal{IM}^{\mathrm{cv}}}(v,s,\pi(s))$ is an concave function. Moreover, 
   \begin{align*}
       \underline{\OF}^{\mathrm{cv}}(v,\pi)(s) = \underline{R}^{\mathrm{cv}}(s,\pi(s)) + \gamma \Omega^{\mathcal{IM}^{\mathrm{cv}}}(v,s,\pi(s))
   \end{align*}
   Since $a\mapsto\underline{R}^{\mathrm{cv}}(s,a)$ is concave, we can deduce that $\pi\mapsto \underline{\OF}^{\mathrm{cv}}(v,\pi)(s)$ is a concave function. Regarding part~\ref{p2:concave}, the fact that $\underline{\OG}(v)\le \underline{\OG}^{\mathrm{cv}}(v)$ follows from definition of $\underline{\OG}$ in~\eqref{eq:IB-2}. Regarding part~\ref{p3:concave}, by Theorem~\ref{thm:pasopt}\ref{p2:pasopt}, the discrete-time dynamical systems~\eqref{eq:valueBE-pes} and~\eqref{eq:valueBE-pes-concave} are monotone and contracting with respect to $\ell_{\infty}$-norm. Note that $\underline{\OG}(v)\le \underline{\OG}^{\mathrm{cv}}(v)$, for every $v\in \real^S$. Therefore, we can use the comparison theorem~\cite[Theorem 3.8.1]{ANM-LH-DL:08-new}, to get $V^*_{\mathrm{p}}\le V^{\mathrm{cv},*}_{\mathrm{p}}$. 
     \end{proof}
     \smallskip
     \begin{remark}The following remarks are in order.
     \begin{enumerate}
         \item \textit{(Computational efficiency)}: using the action-space pessimistic relaxation $\mathcal{IM}^{\mathrm{cv}}$, Theorem~\ref{thm:concave} develops the iteration scheme~\eqref{eq:valueBE-pes-concave} for over-approximating the pessimistic optimal value function of the original IMDP $\mathcal{IM}$. Since the interval Bellman-policy operator is concave in $\pi$, standard convex optimization algorithms (see~\cite{SPB-LV:04}) can be employed to solve the optimization problem at each iteration of~\eqref{eq:valueBE-pes-concave}. 
         
         \item  \textit{(Novelty)}: to the best of our knowledge, \cite{GD-ML-MM-LL:22} is the first paper that proposes to use the concave/convex bounds on the parameters of the IMDPs to approximate its optimal policies. Compared to~\cite{GD-ML-MM-LL:22}, Definition~\ref{def:concaveimdp} and Theorem~\ref{thm:concave} develop a rigorous framework to bound the parameters of the IMDP and provide guarantees for over-approximation of their optimal values. Moreover, our framework is capable of dealing with IMDPs with action-dependent reward functions.  
     \end{enumerate}
     \end{remark}

     \subsection{Iteration-distributed optimization}\label{sub:TDO}

      In practice, estimating the optimal policies using the pessimistic value iterations~\eqref{eq:valueBE-pes-concave} requires solving $|S|$ concave optimization problems with $m$ variables at each iteration step, which can become computationally intractable for IMDPs with large state-space. In this subsection, we consider the pessimistic value iterations~\eqref{eq:valueBE-pes-concave} as the interconnection of a dynamical system described by the value iterations:
      \begin{align}\label{eq:dy-pess}
      v^{k+1} = \underline{\OF}^{\mathrm{cv}}(v^{k},\pi^k),%
      \end{align}
      and an optimization-based feedback controller described by:
      \begin{align}\label{eq:controller-pess}
          \pi^k = \mathrm{argmax}_{\pi\in (A^{\mathrm{cv}})^{S}}\underline{\OF}^{\mathrm{cv}}(v^{k},\pi).
      \end{align}
      Using this system-theoretic perspective toward interval value iterations~\eqref{eq:valueBE-pes-concave}, we propose to implement the optimization-based feedback controller in a distributed fashion. We first need to introduce the following assumption on the IMDPs. 
     
 \begin{assumption}\label{ass:boundreward}
 For the IMDP $\mathcal{IM} = (S,A,[P],[R],\gamma)$,
 \begin{enumerate}
     \item \emph{(Bounded rewards):} there exists $\underline{m}\in \real_{\ge 0}$ such that
     $\sup_{s\in S, a\in A}\underline{R}(s,a) \le \underline{m}$, 
     \item \emph{(Regularity in action variables):} the map $ a \mapsto \underline{R}(s,a)$ is twice continuously differentiable and $c$-strongly concave and $L$-smooth, uniformly in $s\in S$, and the maps 
         \begin{align*}
         &a \mapsto \underline{P}(s',s,a), && a \mapsto \overline{P}(s',s,a),
     \end{align*}
     are twice continuously differentiable and $L$-smooth, for every $s',s\in S$.
 \end{enumerate}
\end{assumption}

Given an IMDP $\mathcal{IM}$ with an action-space pessimistic relaxation $\mathcal{IM}^{\mathrm{cv}}$, we replace the feedback controller described by the optimization problem~\eqref{eq:controller-pess} with $\ell\in \mathbb{Z}_{\ge 0}$ iteration of the projected gradient descent operator $\underline{\OT}^{\mathrm{cv}}:\real^S\times (A^{\mathrm{cv}})^{S}\times \real_{\ge 0}\to \real^S$, 
\begin{align*}
    \underline{\OT}^{\mathrm{cv}}( v,\pi,\beta): = \mathrm{Proj}_{(A^{\mathrm{cv}})^{S}}\left(\pi + \beta \tfrac{\partial }{\partial \pi} \underline{\OF}^{\mathrm{cv}}(v,\pi)\right), 
\end{align*}
where $\beta\ge 0$ is a learning rate.  
The interconnection between the pessimistic value iteration and the pessimistic projected gradient descent operator is shown in Figure~\ref{fig:interconnect}.
\begin{figure}
 \begin{center}
		\includegraphics[width =0.95\linewidth,clip]{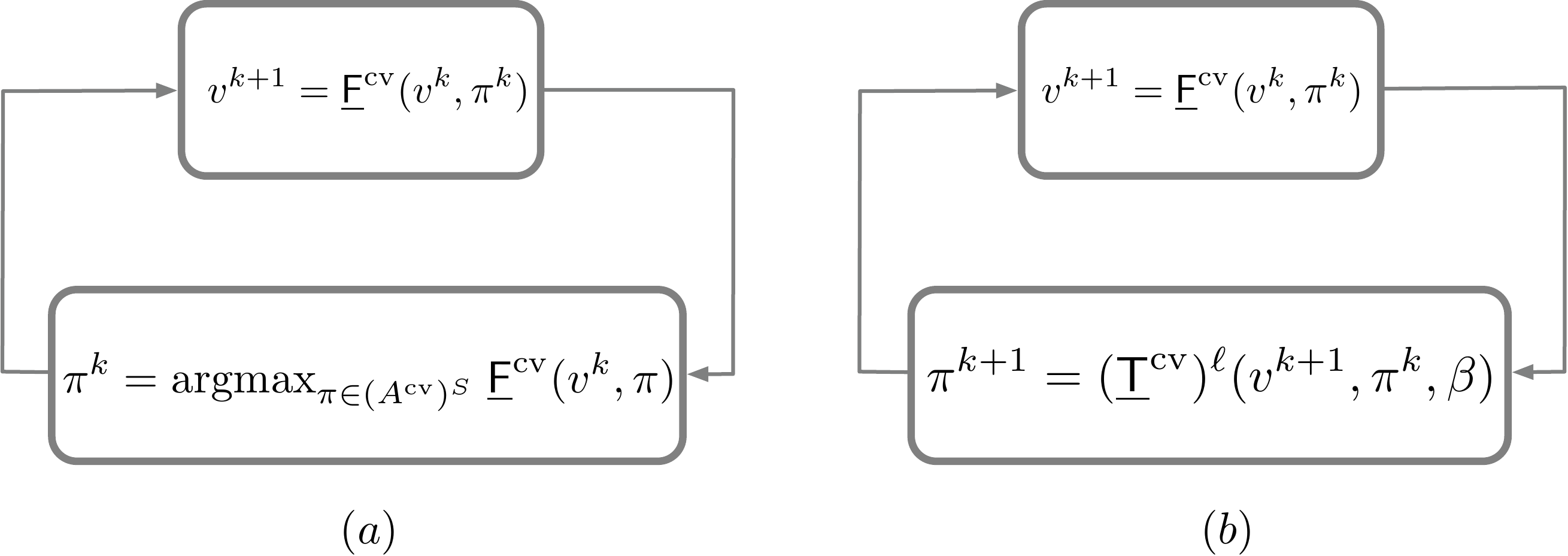}
	\end{center}
	\caption{Value-iterations as a feedback system: (a) is the interconnection between the pessimistic value-iterations and the policy optimization problem, and (b) is the feedback interconnection between the pessimistic value-iterations and the $\ell$ iterations of the projected gradient descent $\underline{\OT}^{\mathrm{cv}}$.}
	\label{fig:interconnect}
\end{figure}
As a result, we define the \emph{pessimistic value-policy iteration} by
\begin{align}\label{eq:iter-pes}
   v^{k+1}&= \underline{\OF}^{\mathrm{cv}}(v^{k},\pi^k),\nonumber \\ \pi^{k+1}&=  (\underline{\OT}^{\mathrm{cv}})^{\ell}( v^{k+1},\pi^k,\beta).
 \end{align}
In order to analyze the pessimistic value-policy iteration, we introduce the map $\pi^*:\real^S\to (A^{\mathrm{cv}})^S\to (A^{\mathrm{cv}})^S$ by 
\begin{align}\label{eq:pstar}
   \pi^*(v) = \mathrm{argmax}_{\sigma\in (A^{\mathrm{cv}})^{S}} \underline{\OF}^{\mathrm{cv}}(v,\sigma). 
\end{align}
The next theorem shows that the interconnection between the pessimistic value iterations and the iteration-distributed optimization described in~\eqref{eq:iter-pes} can be used to approximate the pessimistic optimal value of $\mathcal{IM}^{\mathrm{cv}}$.
 
 \begin{theorem}[Value-policy iterations]\label{thm:distributed}
Consider the IMDP $\mathcal{IM}=(S,A,[P],[R],\gamma)$ with an action-space pessimistic relaxation IMDP $\mathcal{IM}^{\mathrm{cv}} = (S,A^{\mathrm{cv}},[P^{\mathrm{cv}}],[R^{\mathrm{cv}}],\gamma)$ that satisfies Assumption~\ref{ass:boundreward}. Then,
\begin{enumerate}
    \item\label{p1} the compact set $\mathcal{X} = \setdef{(v,\pi)\in \real^{S}_{\ge 0}\times (A^{\mathrm{cv}})^{S}}{v\le \left(\tfrac{\underline{m}}{1-\gamma}\right)\vect{1}_{|S|}}$
    is a forward invariant set for pessimistic value-policy iterations~\eqref{eq:iter-pes}, 
\item\label{p2} let $\{(v^k,\pi^k)\}_{k=0}^{\infty}$ be
    the solution of the pessimistic value-policy iteration~\eqref{eq:iter-pes} starting from  $(v^0=0,\pi^0)\in \mathcal{X}$ with learning rate $\beta = \frac{1}{L}$. Then, for every $k\in \mathbb{Z}_{\ge 1}$,
    \begin{align*}
        v^k \le V^{\mathrm{cv,*}}_{\mathrm{p}} \le v^k + \left(\gamma^k\|v^0-V^{\mathrm{cv,*}}_{\mathrm{p}}\|_{\infty} + \tfrac{(1-\gamma^k)\varepsilon}{1-\gamma}\right)\vect{1}_{|S|},
    \end{align*}
    where $\varepsilon = \sup_{(v,\pi)\in \mathcal{X}}\left\|\tfrac{\partial \underline{\OF}^{\mathrm{cv}}}{\partial \pi}(v,\pi)\right\|_{\infty} (1-\frac{c}{L})^{\ell}\|A^{\mathrm{cv}}\|_{\infty}$.     
    \end{enumerate}

 \end{theorem}
 \smallskip
\begin{proof}
Regarding part~\ref{p1}, it is easy to show that $\mathcal{X}$ is closed and bounded. So it is compact. Now we assume $v^k\in \mathcal{X}$ and we show that $v^{k+1}\in \mathcal{X}$. Using Proposition~\ref{thm:closedform}, for every $s\in S$,  
\begin{align*}
    v^{k+1}(s) & = \underline{R}^{\mathrm{cv}}(s,\pi^{k}(s)) + \gamma \Omega^{\mathcal{IM}^{\mathrm{cv}}}(v^{k},s,\pi^{k}(s)) \\ & \le \underline{m} + \gamma \sum_{i=1}^{j-1} v^k(i_v)\underline{P}^{\mathrm{cv}}(i_v,s,a) \\ & + \gamma \sum_{i=j+1}^{n} v^k(i_v)\overline{P}^{\mathrm{cv}}(i_v,s,a) + \gamma\xi v^k(j_v),
\end{align*}
where $\xi = 1-\sum_{i=1}^{j-1}\underline{P}(i_v,s,a) -\sum_{i=j+1}^{n}\overline{P}(i_v,s,a)$ and $j=\underline{\iota}^{\mathrm{cv}}(v,s,a)$. Using the fact that $v^k(s)\le \frac{\underline{m}}{1-\gamma}$, for every $s\in S$, we get $
    v^{k+1}(s)  \le 
     \underline{m} +\frac{\gamma \underline{m}}{1-\gamma} \le \frac{ \underline{m}}{1-\gamma}$. 
This implies that $v^{k+1}\le \frac{ \underline{m}}{1-\gamma}\vect{1}_{|S|}$. This means that $v^{k+1}\in \mathcal{X}$ and thus $\mathcal{X}$ is a forward invariant set for the discrete-time system~\eqref{eq:iter-pes}. 

Regarding part~\ref{p2}, for every $k\in \mathbb{Z}_{\ge 0}$, 
\begin{align*}
    v^{k+1} = \underline{\OF}^{\mathrm{cv}}(v^k,\pi^k) \le \underline{\OF}^{\mathrm{cv}}(v^k,\pi^*(v^k)) = \underline{\OG}^{\mathrm{cv}}(v^k).
\end{align*}
By Theorem~\ref{thm:IB}\ref{p2:IB}, the map $v\mapsto \underline{\OG}^{\mathrm{cv}}(v)$ is monotone and $V^{\mathrm{cv},*}_{\mathrm{p}}$ is the unique fixed point of $v\mapsto \underline{\OG}^{\mathrm{cv}}(v)$. Moreover, we have $v_0=0\le V^{\mathrm{cv},*}_{\mathrm{p}}$. Thus, by the classical monotone comparison theorem~\cite[Theorem 3.8.1]{ANM-LH-DL:08-new}, we have $v^k\le V^{\mathrm{cv},*}_{\mathrm{p}}$, for every $k\in \mathbb{Z}_{\ge 0}$. We can rewrite the pessimistic value-policy iterations~\eqref{eq:iter-pes} as follows:
\begin{align*}
    v^{k+1} = \underline{\OF}^{\mathrm{cv}}(v^k,\pi^k) &= \underline{\OF}^{\mathrm{cv}}(v^k,\pi^*(v^k)) \\ &+ \int_{0}^{1} \tfrac{\partial \underline{\OF}^{\mathrm{cv}}}{\partial \pi} (v^k,\pi^*(v^k)+ s y^k) y ds.
\end{align*} 
where $y^k := \pi^k - \pi^*(v^k)$, for every $k\in \mathbb{Z}_{\ge 0}$. Thus, 
\begin{align*}
    v^{k+1} - V^{\mathrm{cv},*}_{\mathrm{p}} &= \underline{\OG}^{\mathrm{cv}}(v^k) - \underline{\OG}^{\mathrm{cv}}(V^{\mathrm{cv},*}_{\mathrm{p}}) \\ & + \int_{0}^{1} \tfrac{\partial \underline{\OF}^{\mathrm{cv}}}{\partial \pi} (v^k,\pi^*(v^k)+ s y^k) y ds.
\end{align*}
By Theorem~\ref{thm:IB}\ref{p2:IB}, the map $v\mapsto \underline{\OG}^{\mathrm{cv}}(v)$ is contracting with respect to the $\ell_{\infty}$-norm with rate $\gamma$. This implies that 
\begin{align*}
    \|v^{k+1} - V^{\mathrm{cv},*}_{\mathrm{p}}\|_{\infty} &\le  \gamma \|v^{k} - V^{\mathrm{cv},*}_{\mathrm{p}}\|_{\infty} \\ & + \left\|\int_{0}^{1} \tfrac{\partial \underline{\OF}^{\mathrm{cv}}}{\partial \pi} (v^k,\pi^*(v^k)+ s y^k) y^k ds\right\|_{\infty}.
\end{align*}
Now, we bound the second term on the RHS of the above inequality. Using the triangle inequality, for every $k\in \mathbb{Z}_{\ge 0}$, 
\begin{multline*}
    \left\|\int_{0}^{1} \tfrac{\partial \underline{\OF}^{\mathrm{cv}}}{\partial \pi} (v^k,\pi^*(v^k)+ s y^k) y^k ds\right\|_{\infty} \\ \le \int_{0}^{1} \left\|\tfrac{\partial \underline{\OF}^{\mathrm{cv}}}{\partial \pi} (v^k,\pi^*(v^k)+ s y^k)\right\|_{\infty} \|y^k\|_{\infty} ds \\ \le \sup_{(v,\pi)\in \mathcal{X}}\left\|\tfrac{\partial \underline{\OF}^{\mathrm{cv}}}{\partial \pi}(v,\pi)\right\|_{\infty} \|y^k\|.
\end{multline*}
Note that, for every $s\in S$, we have 
\begin{align}\label{eq:good}
    \tfrac{\partial^2}{\partial \pi^2} \underline{\OF}^{\mathrm{cv}}(v,\pi)(s) \succeq \tfrac{\partial^2}{\partial \pi^2} \underline{R}(s,\pi(s)) \succeq c I_{|S|}, 
\end{align}
where the first inequality in equation~\eqref{eq:good} holds by Proposition~\ref{thm:closedform} and the fact that, for every $i\in \{1,\ldots,j\}$, we have $v(i_v)-v(j_v)\ge 0$ and $a\mapsto \underline{P}^{\mathrm{cv}}(i_v,s,a)$ is concave and, for every $i\in \{j,\ldots,n\}$, we have $v(i_v)-v(j_v)\le 0$ and $a\mapsto \overline{P}^{\mathrm{cv}}(i_v,s,a)$ is convex. The second inequality in equation~\eqref{eq:good} holds by Assumption~\ref{ass:boundreward}. Therefore, equation~\eqref{eq:good} implies that the map $\pi\mapsto \underline{\OF}^{\mathrm{cv}}(v,\pi)(s)$ is $c$-strongly concave. Similarly, one can show that $\pi\mapsto \underline{\OF}^{\mathrm{cv}}(v,\pi)(s)$ is $L$-smooth, for every $s\in S$.  Therefore,
\begin{align*}
  \|y^k\| = \|\pi^{k}-\pi^*(v^{k})\|_{\infty} &\le \left(1-\tfrac{c}{L}\right)^{\ell} \|\pi^{k-1}-\pi^*(v^{k})\|_{\infty} \nonumber\\ &\le \left(1-\tfrac{c}{L}\right)^{\ell} \|A^{\mathrm{cv}}\|_{\infty},
\end{align*}
for every $k\in \mathbb{Z}_{\ge 1}$, where the first inequality holds by the fact that $\pi\mapsto \underline{\OF}^{\mathrm{cv}}(v,\pi)(s)$ is $c$-strongly concave and $L$-smooth  and using the known results about convergence of projected gradient descent~\cite[\S 5.1]{EKR-SB:16}. The second inequality holds since $\pi^{k-1},\pi^{*}(v^{k})\in (A^{\mathrm{cv}})^S$ and $\|(A^{\mathrm{cv}})^S\|_{\infty} = \|A^{\mathrm{cv}}\|_{\infty}$. As a result, for every $k\in \mathbb{Z}_{\ge 1}$
$\|v^{k+1} - V^{\mathrm{cv},*}_{\mathrm{p}}\|_{\infty} \le  \gamma \|v^{k} - V^{\mathrm{cv},*}_{\mathrm{p}}\|_{\infty} + \varepsilon$. This means that, for every $k\in \mathbb{Z}_{\ge 0}$
\begin{align*}
    \|v^{k} - V^{\mathrm{cv},*}_{\mathrm{p}}\|_{\infty} \le  \gamma^k \|v^{0} - V^{\mathrm{cv},*}_{\mathrm{p}}\|_{\infty} + \tfrac{(1-\gamma^k)\varepsilon}{1-\gamma}
\end{align*}
As a result, for every $k\in \mathbb{Z}_{\ge 1}$,  
\begin{align*}
V^{\mathrm{cv},*}_{\mathrm{p}} \le v^{k} + \left(\gamma^k \|v^{0} - V^{\mathrm{cv},*}_{\mathrm{p}}\|_{\infty} + \tfrac{(1-\gamma^k)\varepsilon}{1-\gamma}\right)\vect{1}_{|S|}.
\end{align*} 
This completes the proof of the theorem. 
\end{proof}
\smallskip

\begin{example}[A two-state continuous-action IMDP]\label{ex:simple}
We consider an IMDP $\mathcal{IM}=(S,A,[P],[R],\gamma)$ with two states $S=\{1,2\}$ and the continuous action-space $A=[0,1]\subset\real$ as shown in Figure~\ref{fig:IMDPsimple}. For every $s,s'\in \{1,2\}$, and every $a\in [0,1]$, we define the upper and lower bounds for probability transitions as follows: 
\begin{align*}
    \underline{P}(s',s,a)=0.5a,\qquad \overline{P}(s',s,a)=0.7 + 0.3a.
\end{align*}
For every $a\in [0,1]$, we define the lower and the upper bounds for the reward functions as follows:
\begin{align*}
    \underline{R}(1,a) = 1+4a\sqrt{a}-a^3, \quad \underline{R}(2,a) = 5 - a\sqrt{a}.
\end{align*}
 and we set the discount factor $\gamma=0.9$. For $\mathcal{IM}$, we consider the IMDP $\mathcal{IM}^{\mathrm{cv}} = (S,A^{\mathrm{cv}},[P^{\mathrm{cv}}],[R^{\mathrm{cv}}],\gamma)$ with $A^{\mathrm{cv}}=A=[0,1]$ and with the probability transition bounds $[P^{\mathrm{cv}}]=[P]$. Also the reward bounds are given by
  \begin{align*}
    &\underline{R}^{\mathrm{cv}}(1,a) = 1+4a-a^4, &&\underline{R}^{\mathrm{cv}}(2,a) = 5-a^2.
\end{align*}
The map $a\mapsto \underline{R}^{\mathrm{cv}}(s,a)$ is concave, for every $s\in \{1,2\}$. Moreover, we have $\underline{R}(s,a)\le \underline{R}^{\mathrm{cv}}(s,a)$, for every $s\in \{1,2\}$ and every $a\in [0,1]$. Thus, $\mathcal{IM}^{\mathrm{cv}}$ is an action-space pessimistic relaxation of $\mathcal{IM}$. We use the distributed optimization implementation of the pessimistic value iterations~\eqref{eq:iter-pes} and Theorem~\ref{thm:distributed}\ref{p2} with $\beta = 0.01$, $\ell=1$, and $k=1000$ iterations to obtain $\left[\begin{smallmatrix}35.2000\\
   39.2000\end{smallmatrix}\right]  \le V^{\mathrm{cv},*}_{\mathrm{p}} \le  \left[\begin{smallmatrix}52.7000 \\ 
   56.7000\end{smallmatrix}\right]$.
Using the value iterations~\eqref{eq:dy-pess} with the optimization problem~\eqref{eq:controller-pess}, one can compute $V^{\mathrm{cv},*}_{\mathrm{p}} = \left[\begin{smallmatrix}43.1820 \\ 
   43.8891\end{smallmatrix}\right]$.

\begin{figure}
    \centering
   \includegraphics[width =0.6\linewidth,clip]{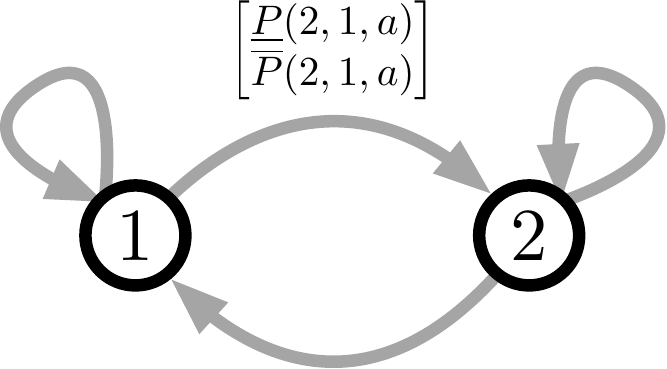}
   \vspace{0.2cm}
    \caption{The state-transition diagram for the interval Markov decision process $\mathcal{IM}$ given in Example~\ref{ex:simple}}
    \label{fig:IMDPsimple}
\vspace{-0.3cm}
\end{figure}

\end{example}

\section{Conclusions}

In this paper, we study IMDPs with continuous action-spaces.
We introduce the pessimistic and the optimistic value iterations for IMDPs and show that they are monotone and contracting dynamical systems. 
Using these observations, we introduce an action-space relaxation of the IMDP and use its value iterations to estimate the optimal policies of the original IMDP. 
Finally, we propose an iteration-distributed implementation of the value iterations and study its convergence to the optimal values.

\addtolength{\textheight}{-12cm}   %

\bibliographystyle{IEEEtran}
\bibliography{SJ.bib}

\end{document}